\documentclass{article}

\usepackage{PRIMEarxiv}

\usepackage[utf8]{inputenc} 
\usepackage[T1]{fontenc}    
\usepackage{hyperref}       
\usepackage{url}            
\usepackage{booktabs}       
\usepackage{amsfonts}       
\usepackage{nicefrac}       
\usepackage{microtype}      
\usepackage{lipsum}
\usepackage{fancyhdr}       
\usepackage{graphicx}       

\usepackage[export]{adjustbox}

\usepackage[english]{babel}

\usepackage{doi}

\usepackage{comment}  

\usepackage{amsmath} 

\usepackage{amssymb} 

\usepackage{amsthm}

\usepackage[ruled,linesnumbered]{algorithm2e}

\usepackage{float}

\usepackage{subcaption}

\usepackage{silence}

\usepackage{dirtytalk}

\newtheorem{theorem}{Theorem}[section]
\newtheorem{lemma}[theorem]{Lemma}
\newtheorem{definition}[theorem]{definition}

\graphicspath{{media/}}     

\title{Slim-ABC: An Optimized Atomic Broadcast Protocol}

\author{ {Nasit S Sony} \\
	University of California, Merced\\
	CA 95340, USA \\
	\texttt{nsony@ucmerced.edu} \\
	\And
	{Xianzhong Ding} \\
	Lawrence Berkeley National Laboratory\\
	CA 94720, USA \\
	\texttt{dingxianzhong@lbl.gov} \\
        \And
	{Mukesh Singhal} \\
	University of California, Merced\\
	CA 95340, USA \\
	\texttt{msinghal@ucmerced.edu} \\
}

\begin{document}
\maketitle

\begin{abstract}
The Byzantine Agreement (BA) problem is a fundamental challenge in distributed systems, focusing on achieving reaching an agreement among parties, some of which may behave maliciously. With the rise of cryptocurrencies, there has been significant interest in developing atomic broadcast protocols, which facilitate agreement on a subset of parties' requests. However, these protocols often come with high communication complexity ($O(ln^2 + \lambda n^3 \log n)$, where $l$ is the bit length of the input, $n$ is the number of parties, and $\lambda$ represents the security parameter bit length). This can lead to inefficiency, especially when the requests across parties exhibit little variation, resulting in unnecessary resource consumption. In this paper, we introduce Slim-ABC, a novel atomic broadcast protocol that eliminates the $O(ln^2 + \lambda n^3 \log n)$ term associated with traditional atomic broadcast protocols. While Slim-ABC reduces the number of accepted requests, it significantly mitigates resource wastage, making it more efficient. The protocol leverages the asynchronous common subset and provable-broadcast mechanisms to achieve a communication complexity of $O(ln^2 + \lambda n^2)$. Despite the trade-off in accepted requests, Slim-ABC maintains robust security by allowing only a fraction ($f+1$) of parties to broadcast requests. We present an extensive efficiency analysis of Slim-ABC, evaluating its performance across key metrics such as message complexity, communication complexity, and time complexity. Additionally, we provide a rigorous security analysis, demonstrating that Slim-ABC satisfies the \textit{agreement}, \textit{validity}, and \textit{totality} properties of the asynchronous common subset protocol.

\end{abstract}


\section{Introduction}\label{sec1}

The Byzantine Agreement (BA) problem is fundamental in distributed systems where multiple computers (parties) must agree on a common value, even if some parties act maliciously or unpredictably \cite{BYZ22, BYZ23}. Achieving agreement in such scenarios is crucial for the reliability and security of distributed systems, especially in asynchronous networks where message delivery times are unpredictable. Traditional BA protocols in synchronous and partially synchronous networks often rely on leader-based approaches, which can suffer from high communication complexity and delays, becoming single points of failure if Byzantine. This issue is more pronounced in large-scale systems like blockchain technologies, where decentralized agreement is essential \cite{BITCOIN01}. To address these issues, asynchronous Byzantine agreement (ABA) protocols are needed. Fischer, Lynch, and Paterson \cite{CONS03} proved that BA protocols do not terminate in asynchronous settings with even one non-Byzantine failure. Ben-Or \cite{BYZ11} showed that introducing randomness allows these protocols to terminate with high probability. Cachin et al. \cite{SECURE02} introduced the ABA, which serves as the foundation for the MVBA protocol \cite{SECURE02}. MVBA allows each party to input a value, with the protocol outputting one party's input, validated by a predefined predicate, using threshold-signature and coin-tossing schemes \cite{THRESH01, BORN01}. However, the protocol has high communication complexity which is $O(ln^2+ \lambda n^2+n^3)$. Recent work by Abraham et al. \cite{BYZ17} and Dumbo-MVBA \cite{BYZ20} reduces this to $O(ln^2+ \lambda n^2)$ using erasure codes. 

Recent efforts belong to the atomic broadcast protocols \cite{HONEYBADGER01, FASTERDUMBO, SPEEDINGDUMBO}, which are built from the asynchronous common subset (ACS) protocol. The ACS protocol is a BA that outputs a subset containing \textit{n-f} input values. The communication complexity of these protocols is $O(ln^2 + \lambda n^3 \log n)$. However, analyzing the protocols reveals that even with threshold encryption, if parties input the same transaction, the outcome resembles the MVBA protocol. We simulated different scenarios to observe the behavior of the protocols and found that honest parties, despite proposing varied requests, might still broadcast the same ones due to differing client request orderings and lack of knowledge of others' requests until an agreement is reached. Consequently, agreeing on a subset of requests does not necessarily improve the total number of accepted requests. 



The above challenges highlight the need for a protocol that has a low communication cost and can output a set of parties' input. The low communication cost can mitigate the effect of having duplicate requests. To find a protocol with low communication cost, we analyze the existing atomic broadcast protocols and find out the key factors that contribute to high communication cost. HoneyBadgerBFT \cite{HONEYBADGER01} provides the first practical Byzantine fault tolerant (BFT) atomic broadcast protocol, and the high communication cost ($O(\lambda n^3 logn)$) of the protocol comes from the use of reliable broadcast protocol (RBC). The RBC protocol ensures the reliability of the message delivery. Fasterdumbo \cite{FASTERDUMBO} also utilizes the RBC protocol for reliable delivery. On the other hand, SpeedingDumbo\cite{SPEEDINGDUMBO} uses a tighter version of RBC, the Provable-Broadcast from Abraham et al. \cite{BYZ17}. The Provable-Broadcast (PB) is an instantiation of verifiable consistent broadcast (VCB) from Cachin et el. \cite{SECURE03}. The PB protocol does not provide a reliable property; therefore, SpeedingDumbo uses a message dissemination and recovery method to recover the message, and it leads the communication cost to $O(\lambda n^3 logn)$. Our main observation is that reducing the number of proposals leads to a more efficient protocol where the probability is high that parties may have duplicate requests. We leverage this reduction technique to design an atomic broadcast protocol. In this protocol, parties agree on a small number of parties' requests ($1 \leq q \leq f+1$), reducing communication complexity to $O(n^2(l+\lambda))$. We randomly select $f+1$ parties to broadcast their requests/proposals, ensuring at least one honest party is included, with an average of two-thirds of the selected parties being honest. If parties agree on one party's request, the communication cost is lower regardless of request variation among selected parties (see Figure \ref{fig:Result1}). If parties agree on $q$ proposals with non-varying requests, the protocol maintains low communication costs (see Figure \ref{fig:Result2}). If requests vary among the $q$ parties, the protocol benefits from both reduced communication costs and an increased number of accepted requests (see Figure \ref{fig:Result3}).

   We propose Slim-ABC, an atomic broadcast protocol designed to have a message and communication complexity like an MVBA protocol, compared to traditional approaches. Slim-ABC leverages a committee selection, prioritized provable-broadcast (pPB) mechanism to reduce the communication complexity to $O(ln^2 + \lambda n^2)$. The primary challenge was to design a protocol that could efficiently output a set of parties' input requests while maintaining low communication costs. We solved this by allowing only a fraction ($f+1$) of parties to broadcast their proposals and by using a threshold encryption scheme to ensure security \cite{HONEYBADGER01}. The second challenge is how to distribute the ($f+1$) parties' proposals among the parties thus they can reach an agreement. To address this challenge, we introduce a new step \textit{suggest} like Sony et al. \cite{PMVBA}. The \textit{suggest} step disseminates the provable-broadcast obtained from pPB in a way that ensures the proposals of the committee members are received by a threshold number of parties. Thus, the parties can reach an agreement on the set of parties' proposals.


At the core of our design is the introduction of a committee from the parties and letting only these parties broadcast their requests. The approach generates proof of broadcast only for selected parties, ensuring that only relevant messages are disseminated efficiently. This selective broadcast mechanism helps to significantly reduce the overall communication complexity. To validate the effectiveness of our proposed protocols, we conducted extensive analysis based on several key metrics: Message Complexity, Communication Complexity, and Time Complexity. Our security and efficiency analysis demonstrate that Slim-ABC significantly reduces communication and message complexities compared to existing atomic broadcast protocols while also maintaining robust security properties.

We summarize our contributions as follows:
\begin{itemize}
    \item \textbf{Slim-ABC Protocol}: We present an atomic broadcast protocol that reduces communication costs by leveraging a committee selection mechanism at the time of proposal broadcast, achieving a communication complexity $O(ln^2 + \lambda n^2)$. The protocol is more efficient when the parties are prone to have duplicate requests. The efficiency is achieved by allowing only a fraction ($f+1$) of parties to broadcast proposals, supported by a threshold encryption scheme for security.
    \item \textbf{Message distribution}: At the core of our protocol, we find and implement a message distribution pattern that efficiently distributes the messages, which significantly reduces the overall communication complexity compared to existing atomic broadcast protocols.
    \item \textbf{Extensive analysis}: We validate our protocols through extensive security and efficiency analysis, demonstrating substantial reductions in communication and message complexities while maintaining robust security properties compared to existing atomic broadcast protocols.
\end{itemize}

The remainder of this paper is organized as follows. Section \ref{Preliminaries} presents the preliminaries, outlining the key concepts and protocols used as foundations for this research. Section \ref{slim-abc} introduces the design of Slim-ABC, detailing each component of the protocol, including committee selection, prioritized provable-broadcast, suggestion, and ABBA-Invocation. In Section \ref{evaluation}, we provide a thorough security and efficiency analysis, demonstrating how Slim-ABC satisfies Byzantine Agreement properties while achieving significant communication and message complexity reductions. We also offer an evaluation of Slim-ABC compared to existing atomic broadcast protocols, focusing on key metrics such as message complexity, communication complexity, and time complexity. Finally, Section \ref{conclusion} concludes the paper by summarizing the contributions and suggesting directions for future work.

\section{Preliminaries}\label{Preliminaries}

\subsection{Definitions and Assumptions}

\subsubsection{Provable-Broadcast} Provable-Broadcast for the selected parties ensures the following properties with negligible probability:
\begin{itemize}
    \item \textbf{PB-Integrity:} An honest party delivers a message at most once.
    \item \textbf{PB-Validity:} If an honest party $p_i$ delivers $m$, then $EX-PB-VAL_i\langle \text{id}, m \rangle = \text{true}$.
    \item \textbf{PB-Abandon-ability:} An honest party does not deliver any message after it invokes PB-abandon(ID).
    \item \textbf{PB-Provability:} For two values $v$, $v'$, if a sender can produce two threshold-signatures $\sigma$, $\sigma'$ such that $\text{threshold-validate}(\langle \text{id}, v \rangle, \sigma) = \text{true}$, then $\text{threshold-validate}(\langle \text{id}, v' \rangle, \sigma') = \text{true}$. This implies that $v = v'$ and at least $f+1$ honest parties delivered a message $m$ such that $m.v = v$.
    \item \textbf{PB-Termination:} If the sender is honest, no honest party invokes PB-abandon(ID), all messages among honest parties are delivered, and the message $m$ that is being broadcast is externally valid, then (i) all honest parties deliver $m$, and (ii) PB(ID, $m$) returns (to the sender) $\sigma$, which satisfies $\text{threshold-validate}(\langle \text{ID}, m.v \rangle, \sigma) = \text{true}$.
    \item \textbf{PB-Selected:} If an honest party $p_i$ delivers $m$, then $m$ is proposed by a selected party.
\end{itemize}

\subsubsection{Cryptographic Abstractions}
Since we aim to design a distributed algorithm in authenticated settings where we use robust, non-interactive threshold signatures to authenticate messages, a threshold coin-tossing protocol to select parties randomly, and a threshold encryption scheme to encrypt messages \cite{SECURE02, PMVBA}, we introduce each of the schemes here.

\begin{enumerate}
    \item \textbf{Threshold Signature Scheme:}
     We utilize the threshold signature scheme introduced in \cite{THRESH01,SECURE02}. The main idea is that there are \( n \) parties, up to \( f \) of which may be faulty. Each party holds a share of a secret key of a signature scheme and can generate a share of a signature on an individual message. \( t \) signature shares are both necessary and sufficient to construct a threshold signature where \( f < t \leq (n-f) \). The threshold signature scheme also provides a public key \( pk \) along with secret key shares \( sk_1, \ldots, sk_n \), a global verification key \( vk \) to verify the message signed by public key \( pk \), and local verification keys \( vk_1, \ldots, vk_n \). Initially, a party \( p_i \) has information on the public key \( pk \), global verification key \( vk \), a secret key share\( sk_i \), and the verification keys for all the parties' secret keys. We describe the security properties of the scheme and related algorithms in Appendix \ref{TSS}.
    
    \item \textbf{Threshold Coin-Tossing Scheme:}
    In the threshold coin-tossing scheme, introduced in \cite{THRESH01,SECURE02}, each party holds a share of a pseudorandom function \( F \). The pseudorandom function \( F \) maps a coin named \( C \) (an arbitrary bit string). A distributed pseudorandom function is a coin that simultaneously produces \( k'' \) random bits. The name \( C \) (the arbitrary bit string) is necessary and sufficient to construct the value \( F(C) \in \{0,1\}^{k''} \) of the particular coin. The parties may generate shares of a coin — \( t \) coin shares are both necessary and sufficient to toss the coin where \( f < t \leq n-f \), similar to threshold signatures. The generation and verification of coin-shares are also non-interactive. We describe the security properties of the scheme and related algorithms in Appendix \ref{TCT}.

   \item{\textbf{Threshold encryption scheme}} \label{app:TPKE}
    A threshold encryption scheme allows any party to encrypt a message to a given public key such that a threshold number of honest parties are required to participate to decrypt the message. The threshold number is $ f+1 $ ( $ 3f+1$ is the total number of parties), and if these $ f+1 $ number of parties compute and reveal decryption shares for an encrypted message, the message can be recovered. Therefore, the adversary is unable to learn about the message until one honest party reveals its decryption share. A threshold scheme provides the following interface:

    \begin{itemize}
       \item TPKE.Setup($1^K$) $\rightarrow$ PK $\{SK_i\}$ generates a public encryption key PK and the secret keys $\{SK_1,SK_2,...SK_n\}$ for each party.
       \item  TPKE.Enc(PK, m) $\rightarrow$ $C$ encrypts a message $m$.
       \item TPKE.DecShare($, C$)$\rightarrow$ $\sigma_i$ produces the $i^{th}$ share of the decryption (or $\bot$ if $C$ is malformed).
       \item  TPKE.Dec(PK, $C, \{i, \sigma_i\}$) $\rightarrow$ $m$ combines a set of decryption share $\{i, \sigma_i\}$ from at least f+1 parties obtain the plaintext m (or, if $C$ contains invalid shares, then the invalid shares are identified.)
    \end{itemize}

\end{enumerate}

Like HB-BFT \cite{HONEYBADGER01} protocol, we use the same threshold encryption scheme of Baek and Zheng \cite{THRESH05}.

\subsubsection{{$(1, \kappa, \epsilon)$}- Committee Selection} A Committee Selection (CS) protocol is executed among $n$ parties (identified from 1 through $n$). If at least $f+1$ honest parties participate, the protocol terminates with honest parties outputting a $\kappa$-sized committee set $C$ such that at least one of $C$ is an honest party. The detailed properties are provided below.

The protocol satisfies the following properties except with negligible probability in cryptographic security parameter $\kappa$:

\begin{itemize}
    \item \textbf{Termination.} If $\langle f+1 \rangle$ honest parties participate in committee selection and the adversary delivers the messages, then honest parties output a set $C$. 
    \item \textbf{Agreement.} Any two honest parties output the same set $C$. 
    \item \textbf{Validity.} If any honest party outputs set $C$, then (i) $|C| = \kappa$, (ii) The probability of every party $p_i \in C$ is same, and (iii) $C$ contains at least one honest party with probability $1-\epsilon$.
    \item \textbf{Unpredictability.} The probability of the adversary to predict the returned committee before an honest party participates is at most $\frac{1}{^nC_\kappa}$.
\end{itemize}

\subsection{System Model}
We assume an asynchronous message-passing system \cite{BYZ17, FASTERDUMBO, HONEYBADGER01}, which consists of a fixed set of parties ($n$). 

In this subsection, we introduce the computation and communication model the adversarial system uses.

\subsubsection{Computation} The model uses standard modern cryptographic assumptions and definitions from \cite{SECURE02, SECURE03}. We model the system modules’ computations as probabilistic Turing machines and provide infeasible problems to the adversary, making it unable to solve the problem. A problem is defined as infeasible if any polynomial-time probabilistic algorithm solves it only with negligible probability. Since the computation modules are probabilistic Turing machines, the adversary uses a probabilistic polynomial-time algorithm. However, given the definition of an infeasible problem, the probability of solving at least one such problem out of a polynomial in $k$ number of problems is negligible. Therefore, we bound the total number of parties $n$ by a polynomial in $k$.

\subsubsection{Communications} We consider an asynchronous network, where communication is point-to-point, and the medium is reliable and authenticated \cite{BYZ20, BYZ30}. Reliability ensures that if an honest party sends a message to another honest party, the adversary can only determine the delivery time but cannot read, drop, or modify the messages. An authenticated medium ensures that if party $p_i$ receives a message $m$, then party $p_j$ sent the message $m$ before party $p_i$ received it.

\subsection{Design Goal}


We aim to design an atomic broadcast protocol named Slim-ABC that reaches an agreement on a subset of parties' requests instead of $n$. To design the Slim-ABC protocol, we utilize a variation of the asynchronous common subset protocol. Here we provide the properties of the atomic broadcast protocol and the validated asynchronous common subset problem. 

\subsubsection{Atomic Broadcast}  An atomic broadcast protocol satisfies the following properties:
\begin{itemize}
    \item \textbf{Agreement:} If an honest party outputs a value $v$, then every honest party outputs $v$.
    \item \textbf{Total Order:} If two honest parties output sequences of values $\langle v_1, v_2, \ldots, v_i \rangle$ and $\langle v_1', v_2', \ldots, v_{i'}' \rangle$, then $v_j = v_j'$ for $j \leq \min(i, i')$.
    \item \textbf{Censorship Resilience:} If a value $v$ is input to $\langle n-f \rangle$ honest parties, then every honest party eventually outputs $v$.
\end{itemize}


\subsubsection{Asynchronous Common Subset (ACS)} An ACS protocol ensures that each party outputs a common subset of all the parties' input. Since we allow input only from $ f+1 $ parties, we modify the definition of the classic ACS protocol.
\begin{definition}[Validated Asynchronous Common Subset (VACS)]
A protocol solves the ACS problem with input from a subset of parties if it satisfies the following conditions except with a negligible probability:
\begin{itemize}
    \item \textbf{Agreement:} If an honest party outputs a set $V$, then every honest party outputs the set $V$.
    \item \textbf{Validity:} If an honest party outputs $V$, then $|V| \geq 1$ and $V$ contains the inputs that satisfy the $\text{externally-valid}( v, \sigma ) = \text{true}$ condition.
    \item \textbf{Totality:} If the selected parties have an input, then all the selected parties can produce an output.
\end{itemize}
\end{definition}

HoneyBadgerBFT \cite{HONEYBADGER01} provides a conversion from ACS to atomic broadcast by adding threshold encryption, and FasterDumbo \cite{FASTERDUMBO} also uses the same conversion. Our work follows the same conversion but differs in that we allow a subset of parties to propose their requests and use an external-validity predicate to validate a value. Therefore, our validity property ensures that the output set $V$ contains at least one value that passes the external-validity condition. For the totality property, we show that if one honest party inputs, then every honest party outputs. See Appendix \ref{ABC :ACS} for the conversion of atomic broadcast from ACS.

\section{Design of Slim-ABC} \label{slim-abc}
\begin{figure*}[h]
    \centering
    \includegraphics[width=0.9\textwidth]{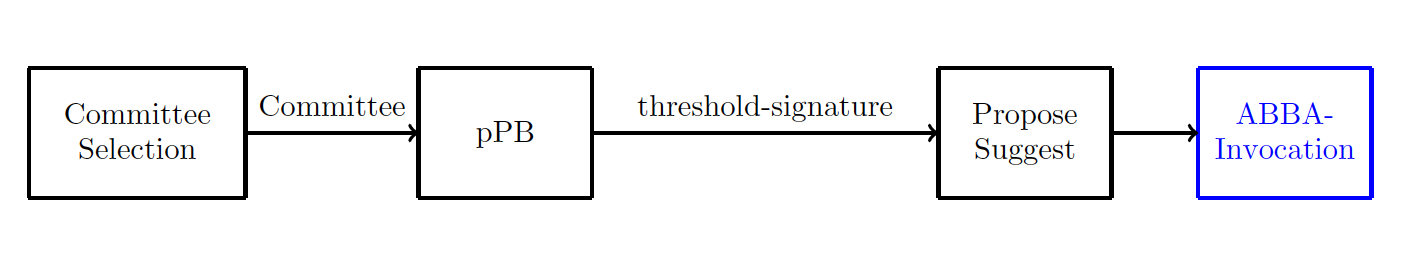}
    \caption{An overview of Slim-ABC.}
    \label{fig:SlimACS}
\end{figure*}


\subsection{Slim-ABC Overview}
This section presents the key components of Slim-ABC. The protocol is composed of four distinct sub-protocols: Committee Selection (CS), Prioritized Provable Broadcast (pPB), Suggestion, and ABBA-Invocation. Honest parties first participate in the Committee Selection process, where a committee of size $f+1$ is formed. Each selected party promotes its request using the pPB protocol, generating a threshold signature as proof of the broadcast. Once a selected party proposes its requests, other parties, upon receiving the proposal, broadcast it as a suggestion. When a party receives a suggestion, it inputs 1 into the corresponding instance of the Asynchronous Binary Byzantine Agreement (ABBA) protocol, referred to as ABBA-Invocation. The black components are our contribution, and the blue ones are adopted from prior work. The framework of the Slim-ABC protocol is depicted in Figure \ref{fig:SlimACS}.

\subsection{\textbf{Committee selection protocol.}} The Committee Selection Protocol is an essential component of the Slim-ABC protocol, designed to reduce communication complexity by selecting a smaller set of parties to broadcast requests rather than involving all $n$ parties. This targeted selection plays a critical role in enhancing efficiency without compromising the security of the protocol. The subset of $f+1$ parties selected at each instance is responsible for performing the agreement task, ensuring the protocol’s progress while maintaining robust security properties. The Committee Selection (CS) protocol is based on a cryptographic coin-tossing scheme, a widely used method in secure distributed systems (e.g., FasterDumbo \cite{FASTERDUMBO}). We follow a similar approach to Sony et al. \cite{PMVBA}, dynamically and randomly selecting $\kappa = f+1$ parties for each instance of the protocol. This guarantees that at least one honest party is included in the committee, and two-thirds of the selected members are expected to be honest. The dynamic selection of the committee also minimizes the risk of adversarial corruption, starvation, and Denial-of-Service (DoS) attacks, ensuring that participation remains fair and secure across all parties.

The CS protocol is illustrated in Algorithm \ref{algo:cs} and involves the following steps:

\begin{itemize}
    \item \textit{ Coin-Share generation:} When SelectCommittee is invoked, a party generates a coin-share $\sigma_i$ for the current instance and broadcasts it to all parties. The party then waits to receive at least $f+1$ coin-shares from other parties (lines 3-5).
    \item \textit{Coin-Share verification:} Upon receiving a coin-share from a party $p_k$ for the first time, the party verifies the authenticity of the coin-share. Valid shares are accumulated in a set $\Sigma$ until $f+1$ valid shares are collected (lines 8-10).
    \item \textit{Committee Selection:} Once a party has received $f+1$ valid coin-shares, it uses the CToss function, which takes the collected coin-shares and the pseudorandom function $F$ as inputs, to randomly select $f+1$ parties to form the committee (lines 6-7).
\end{itemize}
\begin{algorithm}[hbt!]
\LinesNumbered
\DontPrintSemicolon
\SetAlgoNoEnd
\SetAlgoNoLine

\SetKwProg{LV}{Local variables initialization:}{}{}
\LV{}{
   $\Sigma \leftarrow \{\}$\;
}

\SetKwProg{un}{upon}{ do}{}
\un{$SelectCommittee ( id, instance) $ invocation}
{
  $\sigma _i$ $\leftarrow$ $CShare_{id}( r_{id} )$ \;
  \textbf{multi-cast} $( SHARE, id, \sigma _i, instance )$\;
  \textbf{wait until} $|\Sigma| = f+1$\;
  
  \KwRet $CToss ( r_{id}, \Sigma)$\;
}
\SetKwProg{un}{upon receiving}{ do}{}
\un{$ ( SHARE, k, \sigma _k, instance)$ from a party $p_{k}$ for the first time}
{
 \uIf{$CShareVerify( r_{k}, \sigma _k) = true$ }{
    $\Sigma \leftarrow {\sigma _k \cup \Sigma}$}
}

\caption{Committee - Selection: Protocol for party  $p_i$}
\label{algo:cs}
\end{algorithm}

\subsection{Prioritized Provable Broadcast (pPB)\label{pPB-SLIM}}

After the Committee Selection protocol designates the committee members, each selected member must provide a verifiable proof of their proposal to ensure that it has been broadcast to at least $f+1$ honest parties. The input for this protocol includes the ID, requests, and the selected parties, while the output is a threshold signature—a verifiable proof that the same request has been sent to at least $f+1$ honest parties. This proof is essential for maintaining the integrity and consistency of the protocol, as it guarantees that the proposal has been correctly disseminated among the parties. Traditionally, the Verifiable Consistent Broadcast (VCBC) protocol is used to generate such proofs, ensuring that each party can provide a verifiable record of their broadcast proposals. Provable-Broadcast is instantiated from the VCBC protocol by Abraham et al. \cite{BYZ17}. However, since Slim-ABC restricts broadcasting to the selected committee members, we employ a slightly modified version of the Provable-Broadcast protocol, which we term pPB (Prioritized Provable Broadcast).

The pPB protocol is designed to work seamlessly with the selective broadcasting approach established by the Committee Selection process. This adaptation ensures that when a party receives a provable proof from a committee member, no additional verification of the sender’s role is required, as the protocol inherently guarantees it. This mechanism simplifies verification, reducing unnecessary checks and preserving the efficiency introduced by the Committee Selection. The construction of the pPB protocol is detailed in Algorithm \ref{algo:pPB}, and its interactions are illustrated in Figure \ref{fig:pPB}, showcasing its key steps and the role it plays in the broader Slim-ABC protocol.



\begin{figure}[ht]
    \centering
    \begin{minipage}[b]{0.47\textwidth}
        \centering
        \includegraphics[width=\textwidth]{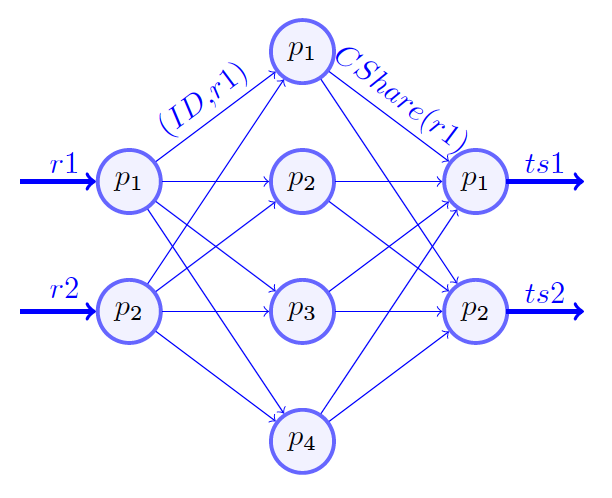}
        \caption{pPB illustration. Here the parties $p_1$ and $p_2$ are the committee members. They first broadcast a message of the form $(ID, r1)$ to every party. When a party receives the message, adds the sign-share $CShare(r1)$ on the message and returns to the sender. A committee member wait for the sign-shares and combines the sign-share to get a threshold-signature ($ts$).}
        \label{fig:pPB}
    \end{minipage}
    \hfill
    \begin{minipage}[b]{0.47\textwidth}
        \centering
        \includegraphics[width=\textwidth]{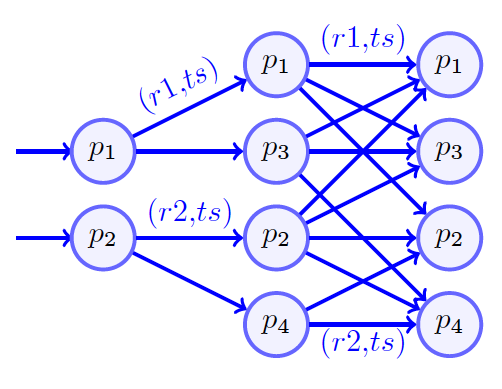}
        \caption{Propose-Suggest illustration. Here the committee members $p_1$ and $p_2$ get their proofs $(ts1, ts2)$ and broadcast that as a proposal to every party. When a party receives a proposal with proof, the party broadcasts the the proposal as a suggestion to every parties (second to third column). A party waits for $2f+1$ suggestions before concluding the steps.}
        \label{fig:PS}
    \end{minipage}
\end{figure}

Here is the construction of the pPB protocol:
\begin{itemize}
    \item Upon the invocation of a $pPB\langle ID, requests, PParties \rangle$ protocol, a party broadcasts the message $(ID, requests)$ to every party. (line 04)
    \item Upon receiving the message $(ID, requests)$ from a party $p_j$ for the first time, a party checks whether the sender is a selected party. If the sender is a selected party, the party adds its sign-share $SigShare$ to the requests and replies to the sender. (lines 12-15)
    \item Upon receiving a sign-share $\sigma_k$ from a party $p_k$, a party verifies the sign-share $\sigma_k$. Then the party adds the sign-share $\sigma_k$ to its set $\Sigma$. (lines 08-10)
    \item A selected party waits for $\langle n-f \rangle$ valid sign-shares. These sign-shares are required to use the $CombineShare$ function to generate a threshold signature. Once the $\langle n-f \rangle$ sign-shares are collected, the threshold signature is returned to the caller. (lines 05-06)
\end{itemize}

\begin{algorithm}[hbt!]
\LinesNumbered
\DontPrintSemicolon
\SetAlgoNoEnd
\SetAlgoNoLine

\SetKwProg{LV}{Local variables initialization:}{}{}
\LV{}{
   $\Sigma \leftarrow \{\}$\;
}

\SetKwProg{un}{upon}{ do}{}
\un{pPB$\langle ID, requests, PParties\rangle$ invocation} 
{
     \textbf{multi-cast} $\langle ID, requests \rangle$\;
     \textbf{wait until} $|\Sigma| = n-f$\;
     \KwRet $\rho \leftarrow CombineShare_{id}\langle requests, \Sigma \rangle$\;
  
}
\;

\SetKwProg{un}{upon}{ do}{}
\un{receiving$\langle requests, \sigma_{k}\rangle $ from the party $p_{k}$ for the first time} {
\uIf{$VerifyShare_{k}\langle requests, (k, \sigma_k)  \rangle$ }{
    $\Sigma \leftarrow {\sigma_k \cup \Sigma}$
 }
}
\;
\un{receiving $\langle ID, requests\rangle $ from the party $p_{j}$ for the first time} {
\uIf{$p_j \in PParties$ }{
   $\sigma_{id} \leftarrow SigShare_{id} \langle sk_{id}, requests \rangle$\;
   $reply \langle requests, \sigma_{id} \rangle$\;
 }
}

\caption{pPB: Protocol for party $p_i$}
\label{algo:pPB}
\end{algorithm}
\subsection{Suggestion}
Following the Prioritized Provable Broadcast (pPB) step, where committee members broadcast their proposals and gather threshold-signatures as proof of dissemination, the \textit{suggest} step ensures efficient communication and moves the protocol toward an agreement. In the Slim-ABC protocol, the \textit{suggest} step removes the need for the costly Reliable Broadcast (RBC) protocols used in traditional atomic broadcast systems like HoneyBadger \cite{HONEYBADGER01} and FasterDumbo \cite{FASTERDUMBO}, as well as eliminating the message dispersal and recovery mechanisms required in SpeedingDumbo \cite{SPEEDINGDUMBO}. Where the pPB step collects the necessary proof for each broadcast request from committee members, the \textit{suggest} step simplifies the process by broadcasting the proposal along with the collected threshold-signatures directly to all parties. This streamlined approach reduces the communication complexity from the $O(n^2)$ messages typical of traditional protocols to $O(n)$ messages in Slim-ABC. The output of the \textit{suggestion} step is a list of threshold-signatures gathered from at least $n-f$ suggestion messages. These proofs are critical, as they guarantee that the proposal has been received and verified by a majority of the parties, pushing the protocol closer to reaching an agreement. The \textit{suggest} step follows the \textit{propose} step where the selected parties broadcast their requests/proposals and proof (threshold-siganture). The visual representation of the two steps is depicted in Figure \ref{fig:PS} and the construction of the two steps are depicted in Algorithm \ref{algo:Slim-ABC} (lines 15-18, lines 29-41). 

\subsection{ABBA-Invocation}

The final phase of the Slim-ABC protocol is the ABBA-Invocation protocol, responsible for reaching an agreement on a proposal submitted by a committee member. After each suggestion-type message is received in the previous step, the party checks whether the corresponding ABBA has been invoked. If it has not been invoked yet, the party inputs $1$ into the ABBA instance. The inputs to this protocol include the proposal's ID, a bit indicating the input as $1$, the corresponding message $m$, the provable threshold signature $\rho$, and the committee member’s ID. The output of the protocol is the set of requests proposed by a committee member, which are then agreed upon by all the parties.

In the ABBA-Invocation step, an asynchronous binary Byzantine agreement protocol biased towards $1$ is employed, enabling efficient agreement on proposals. The process begins with disseminating the vote to all parties using a $V$-type message and collecting votes from other parties (see lines 5-12 and 27-30 in Algorithm \ref{algo:ABBA-Invocation}). This ensures that if $f+1$ honest parties vote $1$, then all honest parties will eventually vote $1$, resulting in an agreement on the corresponding committee member’s proposal (lines 14-17 in Algorithm \ref{algo:ABBA-Invocation}).  After the votes are disseminated, a party invokes the ABBA instance (line 18 of Algorithm \ref{algo:ABBA-Invocation}). If the ABBA reaches an agreement on the request, it returns the request, the threshold signature ($tsign$), and a corresponding bit $b$ equal to $1$. Upon returning from the ABBA instance, the party checks whether $b = 1$. If so, the party verifies whether it already has the corresponding ciphertext. If not, it retrieves the ciphertext $m$ using the threshold signature ($tsign$) from another party (lines 20-21). Since $m$ is encrypted, the party needs to decrypt it by multicasting a decryption share request and collecting $f+1$ valid decryption shares (lines 22-25). Once the decryption is complete, all parties agree on the outcome of the particular instance. The detailed construction of the ABBA-Invocation protocol is provided in Algorithm \ref{algo:ABBA-Invocation}. This step is crucial to ensure that all parties reach an agreement on the submitted proposals, maintaining the integrity of the Slim-ABC protocol.

\begin{algorithm}[hbt!]
\DontPrintSemicolon
\SetAlgoNoEnd
\SetAlgoNoLine

\SetKwProg{un}{upon}{ do}{}
$msg \leftarrow (\bot, \bot)$\;
$u\leftarrow 0$\;\;

\un{invocation of the ABBA-Invocation$(ID, bit, m, \rho, l$) }{

    \uIf{bit = 1}{
        $u\leftarrow 1$\;
        $msg \leftarrow (m, \rho)$\;
        $\textbf{multi-cast}(ID, V, l, u, msg)$\;
        
    }\uElse{
      $u\leftarrow 0$\;
        $msg \leftarrow (m, \rho)$\;
        $\textbf{multi-cast}(ID, V, l, u, msg)$\;
    }
    \textbf{wait until} $\Sigma= 2f+1$\;
    \uIf{ u = 1 }{
      $v\leftarrow (1, (msg))$
    }\uElse{
      $v\leftarrow (0, (msg))$ 
    }

     $(b, tsign) \leftarrow ABBA_l(v)$\;
     $m \leftarrow msg[[1]$\;
    \If{$b$ = 1}{
        \If{m = $\bot$}{
          use $tsign$ to complete the verifiable authenticated broadcast and deliver the ciphertext $m$.  See Appendix \ref{VCBC}
        }
        $decShare \leftarrow TPKE.DecShare( SK_i,m )$\;
         \textbf{multi-cast} $( ID, decShare )$\;
        \textbf{wait for} $f+1$ valid decShare\;
        $msg \leftarrow TPKE.Dec( PK, m, \{i,decShare\})$\;
        \textbf{return} $msg$\;
    }
}

\un{receiving  $(ID, V, l, u', msg') $}{
  \uIf{$u'=1$}{
   $u\leftarrow 1$\;
   $msg \leftarrow msg'$\;
  }
}
\caption{ABBA-Invocation: protocol for the party $p_i$ for an instance $instance$}
\label{algo:ABBA-Invocation}
\end{algorithm}
\subsection{Integration of Subprotocols}

The Slim-ABC protocol reaches an agreement on a subset of parties' requests through a sequence of interconnected sub-protocols. The agreement process begins with the setup of the threshold encryption scheme (see line 1 of Algorithm \ref{algo:Slim-ABC}). Each instance of the protocol starts with the Committee Selection (selectCommittee) protocol, where parties dynamically and randomly select committee members (Algorithm \ref{algo:Slim-ABC}, line 6). Once the committee members are selected, they broadcast their requests using the Prioritized Provable Broadcast (pPB) protocol. This step ensures that each selected party has broadcast the same request to at least $f+1$ honest parties, and these broadcasts are provable (Algorithm \ref{algo:Slim-ABC}, line 10). Upon successful completion of the pPB protocol, the selected committee member broadcasts the proof as a $PROPOSAL$ and as a $SUGGESTION$ (lines 11-14).  If a party is not a committee member, it waits for either a $PROPOSAL$ or a $SUGGESTION$ message. Upon receiving such a message, if no $SUGGESTION$ has been sent yet, the party broadcasts a $SUGGESTION$ message and waits for $2f+1$ suggestions (lines 26-29). When a party receives $2f+1$ suggestions, it checks whether it has already given input to the ABBA instance for all the prioritized parties. If it hasn't, the party initiates the remaining steps of the ABBA-Invocation protocol (lines 20-24 of Algorithm \ref{algo:Slim-ABC}). Each of these sub-protocols—Committee Selection, Prioritized Provable Broadcast, Suggestion, and ABBA-Invocation—ensures that the Slim-ABC protocol operates efficiently and securely, even in the presence of Byzantine faults. This integration allows Slim-ABC to reach an agreement while maintaining low communication complexity and robust fault tolerance.

\begin{algorithm}[hbt!]
\DontPrintSemicolon
\SetAlgoNoEnd
\SetAlgoNoLine

\SetKwProg{un}{upon}{ do}{}
$\langle PK, SK_i \rangle \leftarrow TPKE.Setup( 1^K )$ See \ref{app:TPKE}\;
$instance \leftarrow 1$\;

\While{true}{
$suggest \leftarrow false$\;
$result \leftarrow \{\}$\;
$\Sigma_s \leftarrow 0$\;
$\Sigma \leftarrow  \{\}$ \;
$ PrioritizedParties \leftarrow selectCommittee ( id, instance )$\;

\uIf{$p_{id}$ $\in$ PrioritizedParties}{
    $ID \leftarrow ( instance, id )$\;
    $m \leftarrow TPKE.Enc( PK, requests )$\;
    $\rho \leftarrow pPB ( ID, m, PrioritizedParties)$\;
    \un{$pPB$ return with $\rho$}{
        suggest = true\;
        \textbf{multi-cast} $( PROPOSAL, ID, m, \rho )$\;
        \textbf{multi-cast} $( SUGGESTION, ID, m, \rho, i )$\;   
    }
}\uElse{
    \textbf{wait for} a $PROPOSAL$ or a $SUGGESTION$ type of message\; 
}
\;


\textbf{wait for} $\Sigma_s = 2f+1$\;\;

  \For{$k$ $\in$ $PrioritizedParties$}{   
        \If{no input has been provided to $ABBA_{k}$}{
            $msg \leftarrow ABBA-Invocation(ID, 0, \bot, \bot, k)$\;
            $result \leftarrow result \cup msg$ \;
        }
    }
$instance \leftarrow instance +1$\;
\textbf{output} $result$\;    \;

}
 \un{receiving a $( PROPOSAL, ID', m,  \rho )$ message for the first time}{
        \If{ suggest = false}{
           suggest = true\;
           \textbf{multi-cast} $( SUGGESTION, ID, m, \rho, ID'.id)$\;
        }
       
}
\;

\un{receiving a $( SUGGESTION, ID, m, \rho,l )$ from a selected party $p_j$}{
    $\Sigma_s \leftarrow \Sigma_s + 1$\;
    \If{ suggest = false}{
           suggest = true\;
           \textbf{multi-cast} $( SUGGESTION, ID, m, \rho,l )$\;
    }

    \If{no input has been provided to $ABBA_{l}$}{

        $msg \leftarrow ABBA-Invocation(ID, 1, m, \rho, l)$\;
        $result \leftarrow result \cup msg$
    }
}

\caption{Slim-ABC: protocol for the party $p_i$ for an instance $instance$}
\label{algo:Slim-ABC}
\end{algorithm}

\section{Evaluation}\label{evaluation}


\subsection{Metrics for Evaluation}
We evaluated the performance of our protocols based on the following metrics:
\begin{itemize}
    \item \textbf{Message Complexity:} The total number of messages generated by honest parties during protocol execution.
    \item \textbf{Communication Complexity:} The total bit-length of messages generated by honest parties.
    \item \textbf{Time Complexity:} The total number of rounds of communication required before the protocol terminates.
\end{itemize}
These metrics help us assess the protocol's efficiency, comparing its performance to existing atomic broadcast protocols.

\subsection{Results and Discussion} 

Our analysis demonstrates that the proposed Slim-ABC protocol preserves the key security properties of the Asynchronous Common Subset (ACS) protocol while significantly reducing communication complexity compared to existing atomic broadcast protocols. We provide both security and efficiency analyses to highlight the strengths of Slim-ABC.

\subsubsection{Security Analysis}
The proposed Slim-ABC protocol provides an atomic broadcast protocol for a subset of parties' requests by applying the ACS protocol and the threshold encryption scheme. To analyze the security of Slim-ABC protocol, we considered two main aspects: the reduction from atomic broadcast to ACS and ensuring that the proposed Slim-ABC protocol satisfies the ACS properties. The proposed protocol is a reduction from ACS to prioritized provable broadcast (pPB) and asynchronous binary Byzantine agreement (ABBA) biased towards 1. The ABBA biased towards 1 requires that the provable proof from the pPB protocol must reach at least one honest party or $f+1$ (including $f$ faulty) parties. Lemma \ref{one proposal} and Lemma \ref{common proposal} prove that the proposed protocol satisfies the requirement. Theorem \ref{Theorem}
proves that the protocol satisfies the properties of ABC and ACS protocols. The Lemma \ref{one proposal} was first proposed and proved by Sony et al. \cite{OHBBFT}. We adopt that proof. A version of Lemma \ref{common proposal} is proposed and proved by Sony et al. \cite{PMVBA}.

\begin{lemma}\label{one proposal}
   In the $propose$ step of the protocol, one or more provable-broadcast proof reaches more than one party.
\end{lemma} 

\begin{proof}
   We know that $t \leq f + 1$ parties propose their requests with the proof, and $2f+1 \leq m \leq 3f + 1$ parties receive at least one proposal. Therefore, due to the fraction $\frac{3f+1}{f+1}$, at least one proposal is common to more than one party.
\end{proof} 

\begin{lemma} \label{common proposal}
    In the $suggest$ step, one or more proposals are common to $\langle 2f+1 \rangle$ parties.
\end{lemma} 

\begin{proof}
See Appendix \ref{Lemma2 proof}
\end{proof}

\begin{theorem} \label{Theorem}
Except with negligible probabilities, the Slim-ABC protocol satisfies the \textit{Agreement}, \textit{Validity}, and \textit{Totality} properties of the ACS protocol, given that the underlying prioritized-provable-broadcast, committee-selection, and the $ABBA$ sub-protocols are secure.
\end{theorem}

\begin{proof}
\textit{Agreement:} To prove that the Slim-ABC protocol satisfies the \textit{agreement} property, we prove that when an honest party outputs a set $|V|= m$, then every honest party outputs $V$.

The set $V$ contains the proposal from the $m$ number of committee members, where $1\leq m \leq f+1$. Without the loss of generality, we assume the set $V$ contains one provable broadcast from a selected party. It was received in the propose or suggest step. The corresponding committee member (CM) must receive $1$ for its $ABBA$ instance. Due to the \textit{agreement} property of the $ABBA$ protocol, all honest parties will also output $1$. Hence, a threshold number of honest parties will receive the provable broadcast due to the property of Lemma \ref{common proposal}. 

On the other hand, due to the \textit{validity} property of the $ABBA$ protocol, at least one honest party inputs $1$ to the $ABBA$ instance. This implies that the party must have received the related provable broadcast and message. The \textit{verifiability} property of the pPB protocol ensures that all honest parties will receive the same message (see lines 21-22 of Algorithm \ref{algo:ABBA-Invocation}). 

Hence, every honest party outputs$ \{v_j\}_{j\in CM}= V$


\textit{Validity:} To prove the Slim-ABC satisfies the validity property, we show that $|V| \geq 1$ and $V$ contain the input that satisfies the external-validity property.

If an honest party outputs a set $V= \{v_j\}_{j\in CM}$. We assume the set CM (committee members) includes only one provable-broadcast that was received in the proposal or suggestion step. According to the Slim-ABC protocol, we know that if a   $ABBA$ instance returns $1$, then due to the validity property of $ABBA$, at least one party inputs $1$ to that $ABBA$ instance. It implies that the honest party has received the provable-broadcast and the message.

The \textit{verifiability} property of pPB can ensure that all honest parties will receive $(value, v_j)_{j \in CM}$. Therefore, we have $|V| \geq 1$. Notice that there are at most $f$ faulty parties, and the value satisfies the external-validity property.


\textit{Totality:} To prove that slim-ABC satisfies the \textit{totality} property, we show that all honest parties produce an output if $m$ $(1 \leq m \leq f+1)$ parties have an input.

Since $m$ parties have input, according to the Lemma \ref{common proposal}, at least $f+1$ honest parties can receive value messages from distinct committee members. Besides, according to the CS protocol, at least one honest party belongs to the committee.

We will first prove that at least one $ABBA$ instance returns $1$. (Our assumption is that $m$ is at least $1$)

Let us assume all $ABBA$ instances output $0$. In this case, lines 23-24 of Algorithm \ref{algo:Slim-ABC} will never execute because line 20 implies that it has voted $1$ to at least one $ABBA$ instance; therefore, no $ABBA$ instances get input from an honest party. However, according to the validity property of $ABBA$, which is biased towards $1$, at least $f+1$ honest parties input $0$ to an $ABBA$ instance to output that $ABBA$ instance $0$, which is a contradiction. 

Secondly, since Lemma \ref{one proposal} ensures that a provable-broadcast is common to more than one party and consequently Lemma \ref{common proposal} ensures that at least ($f+1$) parties receive $m$ number of provable-broadcast and input $1$ to those $ABBA$ instances. Again, according to the validity of $ABBA$ those $ABBA$ returns $1$ to all. 

Hence, at least one $ABBA$ instance exists that returns $1$. Due to the validity of $ABBA$ at least one honest party inputs $1$ to $ABBA_k$. It implies that such an honest party must have receives a proposal or suggestion type message and the provable-broadcast. The verifiable property of the pPB protocol now can ensure that all honest parties will have the value (see line 21-12 of ALgorithm \ref{algo:ABBA-Invocation}). Hence all honest parties can produce output for $m$ number of selected parties.

\end{proof}

\subsubsection{Efficiency Analysis}
The efficiency of an atomic broadcast (ABC) protocol depends on message complexity, communication complexity, and running time. We analyze the proposed protocol's efficiency by examining its sub-components: the pPB sub-protocol, committee selection, propose-suggest steps, and the ABBA-Invocation sub-protocol.

\paragraph*{Running Time:} Each sub-protocol and step, except for ABBA-Invocations, has a constant running time. The running time of the proposed protocol is dominated by the ABBA sub-protocols. The Slim-ABC protocol runs the ABBA protocol biased towards 1 $f+1$ times. Therefore the running time of the Slim-ABC protocol is the running time of the $ABBA$ instances, which is $log(f+1)$ or $logn$ \cite{FASTERDUMBO} in expectation.  In conclusion, the expected running time of the protocol is $logn$.

\paragraph*{Message Complexity:} In all sub-protocols and steps, except for pPB and propose steps, each party communicates with all other parties. Every party transmit $O(1)$ information to all other parties (See line 5 of Algorithm \ref{algo:cs}, lines 15-16, 32 and 38 of Algorithm \ref{algo:Slim-ABC} and lines 12 and 23 of Algorithm \ref{algo:ABBA-Invocation}. Each of the multi-cast send $O(1)$ information). Since $n$ parties send $O(1)$ information to the $n$ parties, the message complexity is $O(n^2)$. The expected message complexity of the ABBA protocol is also $O(n^2)$.

\paragraph*{Communication Complexity:} The communication complexity of each sub-protocol and step is $O(n^2(l + \lambda))$, where $l$ is the bit length of input values and $\lambda$ is the bit length of the security parameter. To calculate the communication complexity we use the same approach as message complexity. We observe that in no step a party transmit $O(n)$ information. Thus, the communication complexity is same as message complexity only includes the bit length of the input values and the bit length of the security parameters. The expected communication complexity of the Slim-ABC protocol is also $O(n^2(l + \lambda))$.



\subsection{Comparison with Existing Protocols}
We compared our protocol against the existing atomic broadcast protocols and the other committee based protocols. Our findings indicate:

\subsubsection{Comparison with Existing Atomic Broadcast Protocol} \label{ComparisonABC}
As discussed earlier, when the inputs of each party are nearly identical, outputting the requests of $n-f$ parties is not a viable solution. This approach results in higher computational effort without increasing the number of accepted transactions. Table \ref{table:Table_III} provides a comparison of the communication complexity of our protocol with existing atomic broadcast protocols. Notably, no atomic broadcast protocol can eliminate the multiplication of $O(n^3)$ terms. Here, we focus solely on the communication complexity.

\begin{table}[h!] \label{ABC}
    
\caption{Comparison of the communication complexity with the existing atomic broadcast protocols}
\begin{center}
\begin{tabular}{||c |c |c |c ||} 
 \hline
 Protocols & Communication Complexity  \\ [0.5ex] 
 \hline\hline
 HB-BFT/BEAT0 \cite{HONEYBADGER01} & $O(ln^2 + \lambda n^3 logn)$  \\ 
 \hline
 BEAT1/BEAT2 \cite{BYZ06} & $O(ln^3 + \lambda n^3)$   \\
 \hline
 Dumbo1 \cite{FASTERDUMBO} & $O(ln^2 + \lambda n^3 logn)$   \\
 \hline
 Dumbo2 \cite{FASTERDUMBO} & $O(ln^2 + \lambda n^3 logn)$  \\ [1ex] 
 \hline
 Speeding Dumbo \cite{SPEEDINGDUMBO} & $O(ln^2 + \lambda n^3 logn)$  \\ [1ex] 
 \hline
  Our Work  & $O(ln^2 + \lambda n^2)$  \\ [1ex] 
 \hline
\end{tabular}
\label{table:Table_III}
\end{center}
\end{table}

\subsubsection{Comparison of Resilience, Termination, and Safety with Committee-Based Protocols}
We compare our work with notable committee-based protocols, specifically focusing on resilience, termination, and safety properties. Table \ref{table:Table_II} highlights these comparisons. COINcidence \cite{BYZ19} assumes a trusted setup and violates optimal resilience. It also does not guarantee termination and safety with probability (w.p.) $1$. Algorand \cite{BYZ21} assumes an untrusted setup, with resilience dependent on network conditions, and does not guarantee termination w.p. $1$. The Dumbo \cite{FASTERDUMBO} protocol uses a committee-based approach, but its committee-election protocol does not guarantee the selection of an honest party, thus failing to ensure agreement or termination with probability $1$. Our protocol achieves optimal resilience and guarantees both termination and safety, as our committee-election process ensures the selection of at least one honest party. This guarantees that the protocol can make progress and reach agreement despite adversarial conditions.

\begin{table}[h!]
    
\caption{Comparison for performance metrics of the committee based protocols}
\begin{center}
\begin{tabular}{||c |c |c |c ||} 
 \hline
 Protocols & n$>$ & Termination & Safety\\ [0.4ex] 
 \hline\hline
 COINcidence \cite{BYZ19} & 4.5f & whp & whp\\ 
 \hline
 Algorand \cite{BYZ21} & * & whp & w.p. 1\\ 
 \hline
 Dumbo1 \cite{FASTERDUMBO} & 3f & whp & w.p. 1\\ 
 \hline
 Dumbo2 \cite{FASTERDUMBO} & 3f & whp & w.p. 1\\ 
 \hline
 Our work & 3f & w.p. 1 & w.p. 1\\
 \hline
\hline
\end{tabular}
\label{table:Table_II}
\end{center}
\end{table}


\section{Conclusion}\label{conclusion}
In this paper, we addressed the Byzantine Agreement (BA) problem in designing atomic broadcast protocols, presenting a novel protocol Slim-ABC. This protocol reduces message and communication complexity by utilizing a smaller, randomly selected subset of parties and leveraging a prioritized provable-broadcast mechanism with threshold encryption. Our extensive security and efficiency analysis demonstrate substantial reductions in message and communication complexities compared to the existing atomic broadcast protocols without compromising security. However, the protocol's reliance on random selection introduces performance variability, and their security assumes a majority of honest parties, which may not hold in highly adversarial environments. Future work can focus on increasing committee size to increase the accepted requests without compromising message and communication complexities. Furthermore, we can focus on testing the protocol in real-world systems like blockchain platforms, enhancing their resilience to complex adversarial models, and integrating them with other BA mechanisms to create more efficient and secure distributed systems.

\bibliography{references}

\appendix


\section{Definitions}

\subsection{Verifiable Consistent Broadcast} 
\label{VCBC}

\begin{definition}[Verfiability] A consistent broadcast protocol is called verifiable if the following holds, except with negligible probability: When an honest party has delivered $m$, then it can produce a single protocol message $M$ that it may send to other parties such that any other honest party will deliver $m$, upon receiving $M$.

\end{definition}
A protocol completes a verifiable consistent broadcast if it satisfies the following properties:

\begin{itemize}
    \item \textbf{Validity.} If an honest party sends $m$, then all honest parties eventually delivers $m$.
    \item \textbf{Consistency.} If an honest party delivers $m$ and another honest party delivers $m'$, then $m=m'$.
    \item \textbf{Integrity.} Every honest party delivers at most one request. Moreover, if the sender $p_s$ is honest, then the request was previously sent by $p_s$.
\end{itemize}
\subsection{Asynchronous binary Byzantine Agreement (ABBA)} \label{ABBAD}
The ABBA protocol guarantees the following properties. Additionally, the biased external validity property applies to the biased ABBA protocol.

\begin{itemize}
    \item \textbf{Agreement.} If an honest party outputs a bit $b$, then every honest party outputs the same bit $b$.
    \item \textbf{Termination.} If all honest parties receive input, then all honest parties will output a bit $b$.
    \item \textbf{Validity.} If any honest party outputs a bit $b$, then $b$ was the input of at least one honest party.
    \item \textbf{Biased External Validity.} If at least $\langle f + 1 \rangle$ honest parties propose $1$, then any honest party that terminates will decide on $1$.
\end{itemize}

\subsection{Threshold Signature Scheme} \label{TSS}
We utilize the threshold signature scheme from \cite{THRESH01, SECURE02}. The security properties and the algorithm definitions we use here are adopted from \cite{PMVBA}. The $(f+1, n)$ non-interactive threshold signature scheme provides a set of algorithms used by $n$ parties, with up to $f$ potentially faulty. The scheme satisfies the following security properties, except with negligible probabilities:

\begin{itemize}
    \item \textbf{Non-forgeability.} A party requires total $t$ \textit{signature shares} to output a valid threshold signature. Since an adversary can corrupt up to $f$ parties ($f < t$) and thus cannot generate enough \textit{signature shares} to create a valid threshold signature as a proof for a message, it is computationally \textit{infeasible} for an adversary to produce a valid threshold signature.
    \item \textbf{Robustness.} It is computationally \textit{infeasible} for an adversary to produce $t$ (where $t > f$) valid \textit{signature shares} such that the output of the share combining algorithm is not a valid threshold signature.
\end{itemize}

The scheme provides the following algorithms:
\begin{itemize}
    \item \textit{Key generation algorithm: KeySetup($\{0,1\}^\lambda, n, f+1) \rightarrow \{UPK, PK, SK\}$}. Given a security parameter $\lambda$, this algorithm generates a universal public key $UPK$, a vector of public keys $PK := (pk_1, pk_2, \ldots, pk_n)$, and a vector of secret keys $SK := (sk_1, sk_2, \ldots, sk_n)$.

    \item \textit{Share signing algorithm: SigShare$_i(sk_i, m) \rightarrow \sigma_i$}. Given a message $m$ and a secret key share $sk_i$, this deterministic algorithm outputs a signature share $\sigma_i$.

    \item \textit{Share verification algorithm: VerifyShare$_i(m, (i, \sigma_i)) \rightarrow 0/1$}. This algorithm takes three parameters as input: a message $m$, a signature share $\sigma_i$, and the index $i$. It outputs $1$ or $0$ based on the validity of the signature share $\sigma_i$ (whether $\sigma_i$ was generated by $p_i$ or not). The correctness property of the signing and verification algorithms requires that for a message $m$ and party index $i$, $\Pr[VerifyShare_i(m, (i, SigShare_i(sk_i, m))) = 1] = 1$.

    \item \textit{Share combining algorithm: CombineShare$_i(m, \{(i, \sigma_i)\}_{i \in S}) \rightarrow \sigma / \perp$}. This algorithm takes two inputs: a message $m$ and a list of pairs $\{(i, \sigma_i)\}_{i \in S}$, where $S \subseteq [n]$ and $|S| = f+1$. It outputs either a signature $\sigma$ for the message $m$ or $\perp$ if the list contains any invalid signature share $(i, \sigma_i)$.

    \item \textit{Signature verification algorithm: Verify$_i(m, \sigma) \rightarrow 0/1$}. This algorithm takes two parameters: a message $m$ and a signature $\sigma$, and outputs a bit $b \in \{0, 1\}$ based on the validity of the signature $\sigma$. The correctness property of the combining and verification algorithms requires that for a message $m$, $S \subseteq [n]$, and $|S| = f+1$, $\Pr[\text{Verify}_i(m, \text{Combine}_i(m, \{(i, \sigma_i)\}_{i \in S})) = 1 \mid \forall i \in S, \text{VerifyShare}_i(m, (i, \sigma_i)) = 1] = 1$.

\end{itemize}

\subsection{Threshold Coin-Tossing} \label{TCT}
We utilize the threshold coin-tossing scheme from \cite{THRESH01, SECURE02}. The security properties and the algorithm definitions we use here are adopted from \cite{PMVBA}. We assume a trusted third party has an unpredictable pseudo-random generator (PRG) $G : R \rightarrow \{1, \ldots, n\}^s$, known only to the dealer. The generator takes a string $r \in R$ as input and returns a set $\{S_1, S_2, \ldots, S_s\}$ of size $s$, where $1 \leq S_i \leq n$. Here, $\{r_1, r_2, \ldots, r_n\} \in R$ are shares of a pseudorandom function $F$ that maps the coin name $C$. The threshold coin-tossing scheme satisfies the following security properties, except with negligible probabilities:

\begin{itemize}
    \item \textbf{Pseudorandomness.} The probability that an adversary can predict the output of the $F(C)$ is $\frac{1}{2}$. The adversary interacts with the honest parties to collect \textit{coin-shares} and waits for $t$ \textit{coin-shares}, but to reveal the coin $C$ and the bit $b$, the adversary requires at least $\langle t-f\rangle$ \textit{coin-shares} from the honest parties. If the adversary predicts a bit $b$, then the probability is $\frac{1}{2}$ that $F(C) = b$ ($F(C) \in \{0, 1\}$). Although the description is for single-bit outputs, it can be trivially modified to generate $k$-bit strings by using a $k$-bit hash function to compute the final value.
    \item \textbf{Robustness.} It is computationally \textit{infeasible} for an adversary to produce a coin $C$ and $t$ valid \textit{coin-shares} of $C$ such that the share-combine function does not output $F(C)$.
\end{itemize}

The dealer provides a private function $CShare_i$ to every party $p_i$, and two public functions: $CShareVerify$ and $CToss$. The private function $CShare_i$ generates a share $\sigma_i$ for the party $p_i$. The public function $CShareVerify$ can verify the share. The $CToss$ function returns a unique and pseudorandom set given $f+1$ validated coin shares. The following properties are satisfied except with negligible probability:

\begin{itemize}
    \item For each party $i \in \{1, \ldots, n\}$ and for every string $r_i$, $CShareVerify(r_i, i, \sigma_i) = \text{true}$ if and only if $\sigma_i = CShare_i(r_i)$.
    \item If $p_i$ is honest, then it is impossible for the adversary to compute $CShare_i(r)$.
    \item For every string $r_i$, $CToss(r, \Sigma)$ returns a set if and only if $|\Sigma| \geq f+1$ and each $\sigma \in \Sigma$ and $CShareVerify(r, i, \sigma) = \text{true}$.
\end{itemize}

\section{Agreement protocol}


\subsection{\textbf{Asynchronous Binary Byzantine Agreement (ABBA)}} \label{appendix:ABBA}
The ABBA protocol allows parties to agree on a single bit $b \in \{0, 1\}$ \cite{BYZ10, SECURE05, SIG01}. We have adopted the ABBA protocol from \cite{SECURE02}, as given in Algorithm \ref{algo:ABBA}. The expected running time of the protocol is $O(1)$, and it completes within $O(k)$ rounds with probability $1 - 2^{-k}$. Since the protocol uses a common coin, the total communication complexity becomes $O(kn^2)$. For more information on how to realize a common coin from a threshold signature scheme, we refer interested readers to the \cite{HONEYBADGER01}.


\paragraph*{Construction of the ABBA biased towards 1} 
We use the ABBA protocol from \cite{SECURE02}. We optimize and changed the protocol for biased towards $1$. The biases towards $1$ property ensures that if at least one party input $1$ in the pre-process step. The pseudocode of the ABBA protocol biased towards 1 is given in Algorithm \ref{algo:ABBA}, and a step-by-step description is provided below:

\begin{itemize}
    \item \textbf{Pre-process step} . Generate an $\sigma_0$ share on the message and multi-cast the pre-process type message.
    \item Collect $2f+1$ proper pre-processing messages. (see (Algorithm \ref{algo:ABBA})).

    \item \textbf{Repeat loop:} Repeat the following steps 1-4 for rounds round = 1,2,3,...
    \begin{itemize}
        \item Pre-Vote step. (see Algorithm \ref{algo:ABBA-PreVote})
        \begin{itemize}
            \item If round = 1, $b=1$ if there is a pre-processing vote for $1$ (biased towards 1, taking one vote instead of majority) else  $b=0$. (see lines 3-4).
            \item If round $>$ $1$, if there is a threshold signature on main-vote message from round-1 then decide and return. (see lines 18-20)
            \item Upon receiving main-vote for $0/1$, update $b$ and the justification. (see lines 12-17) 
            \item $b= F(ID, r-1)$, all the main-vote are abstain and the justification is the threshold signature of the abstain vote. (see lines 6-7)
            \item Produce signature-share on the message (ID, pre-vote, round, b) and multicast the message of the form pre-vote,round,b,justification, signature-share). (lines 9-11)
        \end{itemize}
        \item Main-vote step. (See Algorithm \ref{algo:ABBA-MainVote})
        \begin{itemize}
            \item Collect (2f+1) properly justified round pre-vote messages. (lines 14-19)
            \item If there are (2f+1) pre-votes for 0/1, $v=0/1$ and the justification is the threshold-signature of the the sign-shares on pre-vote messages. (lines 5-7)
            
            \item If there are (2f+1) pre-votes for both $0$ and $1$, $v=abstain$ and the justification is the two sign-shares from pre-vote 0 and pre-vote 1. (lines 9-10)
            \item Produce signature-share on the message (ID, pre-vote, round, v) and multi-cast the message of the form (main-vote,round,v,justification, signature-share) (lines 11-13)
        \end{itemize}
        \item  Check for decision. (See Algorithm \ref{algo:ABBA-CheckForDecision})
        \begin{itemize}
            \item Collect (2f+1) properly justified main-votes of the round $round$. (line 3)
            \item If these is no abstain vote, all main-votes for $b\in \{0,1\}$, then decide the value $b$. Produce a threshold signature on the main votes' sign-shares and multi-cast the threshold signatures to all parties and return.  (lines 4-7)
            \item Otherwise, go to Algorithm \ref{algo:ABBA-CheckForDecision}. line (11)
        \end{itemize}
        \item Common Coin. (See Algorithm \ref{algo:ABBA-CommonCoin})
        \begin{itemize}
            \item Generate a coin-share of the coin (ID, round) and send to all parties a message of the form (coin, round,coin-share). (lines 1-4) 
            \item Collect (2f+1) shares of the coin (ID,round $\sigma_k$), and combine these shares to get the value $F(ID, round) \in \{0,1\}$. (lines 5-6)
        \end{itemize}
        
    \end{itemize}
\end{itemize}

\begin{algorithm}[ht!]
\DontPrintSemicolon
\SetAlgoNoEnd
\SetAlgoNoLine
\SetKwProg{un}{upon receiving}{ do}{}

\SetKwProg{ABBA}{upon}{ do}{}

\ABBA{ABBA(m)}{
\tcc{Preprocess Step.}
 $\sigma_0 \leftarrow SigShare_i (sk_i,m_i)$  see \cite{SECURE02}\;
 
 \textbf{multi-cast} $( pre-process, m_i, \sigma_0)$\;
 
 wait until at least $(n-f)$ pre-process messages have been received.\;

\For{$round = 1,2,3,...$}
{
  \textbf{Prevote Step :} Algorithm \ref{algo:ABBA-PreVote}\;

  \textbf{Main-vote Step:} Algorithm \ref{algo:ABBA-MainVote}\;

  \textbf{Check For Decision :} Algorithm \ref{algo:ABBA-CheckForDecision} \;

  \textbf{Common Coin:} Algorithm \ref{algo:ABBA-CommonCoin}\;

}
}
\caption{ABBA biased towards 1: protocol for party $p_i$ }
\label{algo:ABBA}
\end{algorithm}

\begin{algorithm}[ht!]
\DontPrintSemicolon
\SetAlgoNoEnd
\SetAlgoNoLine
\SetKwProg{un}{upon}{ do}{}
$b \leftarrow \perp$\;
$justification \leftarrow  \perp$\;
\uIf{round = 1}{
  $b \leftarrow 1$ if there is any pre-process message with $m=1$ (biased towards $1$), otherwise 0.
}\uElse{
     $b = F(ID, round-1)$\;
    $justification \leftarrow threshold-signature \langle ID, main-vote, round - 1,abstain\rangle$\;
    
    \textbf{ wait for} $n-f$ justified main-vote \;
     

}

$m_i \leftarrow (ID, pre-vote, round, b)$\;
$\sigma \leftarrow SigShare_i (sk_i,m_i)$  \;
 
\textbf{multi-cast} $( pre-vote, r, b, justification, \sigma)$\;

 \un{receiving $\langle  main-vote,round, v, justification, \sigma  \rangle$ for the first time from party $p_k$} {
  \uIf{v = 0}{
     b = 0\;
  }\uElseIf{v = 1}{
     b = 1\;
  }

   $justification \leftarrow threshold-signature \langle ID, pre-vote, round - 1,b\rangle$
 
 }

 \un{receiving $\langle  b, threshold-signature  \rangle$ for the first time from party $p_k$} {
  \textbf{multi-cast}(b, threshold-signature)\;
  \textbf{return}(b, threshold-signature)\;  
 
 }

\caption{ABBA biased towards 1: Pre-vote step}
\label{algo:ABBA-PreVote}
\end{algorithm}

\begin{algorithm}[ht!]

\DontPrintSemicolon
\SetAlgoNoEnd
\SetAlgoNoLine
\SetKwProg{un}{upon}{ do}{}

\tcc{Mainvote Step.}

   $\Sigma$ = \{\}\;
   $PV_0$ = \{\}\;
   $PV_1$ = \{\}\;
   \textbf{wait until} $|\Sigma|$ = 2f+1\;

   \uIf{ $|PV_0| = 2f+1$ or $|PV_1| = 2f+1$}{
   $v \leftarrow 0/1$\;
   $justification \leftarrow CombineShare_i(v,{i, \sigma_i}_{i \in \Sigma})$\;
   }\uElse{
      $v \leftarrow abstain$\;
       $justification \leftarrow (\sigma_i \in PV_0, \sigma_j \in PV_1)$\;
   }

   $m_i \leftarrow (ID, main-vote, round, v)$\;
$\sigma \leftarrow SigShare_i (sk_i,m_i)$  \;
 
 \textbf{multi-cast} $( main-vote, round, v, justification, \sigma)$\;
    

 \un{receiving $\langle  pre-vote, round, b, justification, \sigma  \rangle$ for the first time from party $p_k$} {
\uIf{b = 0}{
   $PV_0 \leftarrow PV_0 +  1 $\;
}\uElseIf{b = 1}{
   $PV_1 \leftarrow PV_1 +  1 $\;
}
 $\Sigma \leftarrow \Sigma +  \sigma $\;
}

\caption{ABBA biased towards 1: Main-Vote step }
\label{algo:ABBA-MainVote}
\end{algorithm}

\begin{algorithm}[H]

\DontPrintSemicolon
\SetAlgoNoEnd
\SetAlgoNoLine
\SetKwProg{un}{upon}{ do}{}


  $\Sigma$ = \{\}\;
  $isAbstain = no$\;
  \textbf{wait until} $|\Sigma|$ = 2f+1\;
  \uIf{isAbstain = no}{
    threshold-signature = $CombineShare_{i}(b, {(i,\sigma_i)}_{i\in \Sigma})$ \;
   \textbf{ multi-cast}(threshold-signature)\;
    return (b, threshold-signature)\;
  }\uElse{
    go to  Algorithm \ref{algo:ABBA-CommonCoin}\;
  }

\un{receiving $\langle main-vote, r, v, justification, \sigma \rangle$ for the first time from party $p_k$} {
\uIf{v = abstain}{
   $isAbstain = yes$\;
}
 $b = v$\;
 $\Sigma \leftarrow \Sigma +  \sigma $\;
}

\caption{ABBA biased towards 1: Check for Decision for party $p_i$ }
\label{algo:ABBA-CheckForDecision}
\end{algorithm}

\begin{algorithm}[H]
\DontPrintSemicolon
\SetAlgoNoEnd
\SetAlgoNoLine
\SetKwProg{un}{upon}{ do}{}

  $\Sigma$ = \{\}\;
  $\sigma_i \leftarrow CShare(r_i)$\;
 \textbf{ multi-cast} $\langle coin, round, \sigma_i\rangle$\;
  \textbf{wait until} $|\Sigma|$ = 2f+1\;
  $F(ID, round) \in \{0,1\} \leftarrow CToss(r, \Sigma)$\;

\un{receiving $\langle coin, round, \sigma_k \rangle$ for the first time from party $p_k$} {

 $\Sigma \leftarrow \Sigma + \sigma_k $\;
}

\caption{ABBA biased towards 1: Common Coin for party $p_i$}
\label{algo:ABBA-CommonCoin}
\end{algorithm}

\section{Miscellaneous}\label{secA1}

\subsection{Atomic broadcast from ACS.} \label{ABC :ACS}
HB-BFT \cite{HONEYBADGER01} protocol achieves atomic broadcast using the ACS protocol and the threshold encryption scheme. In this protocol, every party proposes its transactions, and at the end of the protocol, parties reach an agreement on at least $f+1$ honest parties' proposals. So, parties choose their transaction list randomly, which helps to have varying proposals from multiple parties. However, an adversary can censor the transactions and delay a particular transaction from getting accepted in the log. To prevent this, parties use threshold encryption and decryption techniques that helps to hide any transactions until the parties reach an agreement. We also follow the same threshold encryption scheme to avoid censorship resilience.

\subsection{Differed Proof} \label{Lemma2 proof}

The proof is adopted from \cite{PMVBA}.
\begin{proof} Since the selected parties can be byzantine and the adversary can schedule the message delivery to delay the agreement, we have considered the scenarios below. 
\begin{enumerate}
    \item Among $\langle f+1 \rangle$ selected parties, $f$  parties are non-responsive.
    \item Selected $\langle f+1 \rangle$ parties are responsive, but other $f$ non-selected parties are non-responsive.
    \item Every party is responsive, including the selected $\langle f+1 \rangle$ parties.
    \item Selected $t \leq \langle f+1 \rangle$ parties are responsive, and total $m$ parties are responsive, where $\langle 2f+1 \rangle \leq m \leq n $.
    
\end{enumerate}

We will first prove that the first three scenarios are a special case of scenario four.

\begin{enumerate}
    \item For case 1, $t=1$ and $m = t+2f = 2f+1$. So it is the same as case 4.
    \item For case 2, $t=f+1$ and $m=t+f=2f+1 < 3f+1$. So, it is the same as case 4.
    \item For case 3, $t=f+1$ and $m=t+2f=3f+1=n$. So, it is the same as case 4.
\end{enumerate}

We prove that in every scenario, at least one party's proposal reaches $ 2f+1 $ parties. Since we have proved that case (1), (2) and (3) are the special case of case (4), proving for these case is enough (proving for case (4) covers all case).

\begin{enumerate}

    \item  Since among $ f+1 $ selected parties, $f$ parties are non-responsive, only one party completes the $pPB$ protocol and proposes the provable-broadcast proof. If any party receives a provable-broadcast proof, then the provable-broadcast proof is from the responsive selected party. Since every party receives the provable-broadcast proof for the same party's proposal, the proposal reaches at least $ 2f+1 $ parties.
    
    \item The $ f+1 $ selected parties are responsive and complete the $pPB$ protocol. Each selected party broadcasts the provable-broadcast proof, and $ 2f+1 $ parties receive the provable-broadcast proof ( another $f$ number of parties are non-responsive). Any party receives a provable-broadcast proof, suggests the received provable-broadcast proof, and waits for $ 2f+1 $ suggestions. If a party receives $ 2f + 1 $ suggestions, then these suggestions include all $ f+1 $ parties' provable-broadcast proofs because among the $ 2f + 1 $ received suggestions, $ f+1 $ number of suggestions are from the selected parties. So, every proposal reaches $ 2f + 1 $ parties.

    \item The proof is by contradiction. Let no proposal reach more than $2f$ parties. Since we assume every party is responsive, every party receives a proposal in the $propose$ step. There must be a $( 3f+1)*( 2f+1)$ suggestion messages. If no proposal can be suggested to more than $2f$ parties, the total number of suggestions is $(3f+1) * 2f <  (2f+1)* (2f+1)$ (Though a proposal can be suggested by more than one party we assume that every party suggests to the same $2f$ parties otherwise it would fulfill the requirement of $ 2f+1$ proposals). However, honest parties must send enough suggestion messages to ensure the protocol's progress, and the adversary eventually delivers the messages. Therefore, at least one party's proposal reaches $2f+1$ parties, a contradiction.

    \item The proof is by contradiction. Let no proposal reaches to more than $2f$ parties. If $1 \leq t \leq  f+1 $ parties distribute their verifiable proof to $ 2f+1 \leq m <  3f+1 $ parties and no proposal reaches more than $2f$ parties, then there must be no more than $m * 2f$ suggestions. However, $m$ parties must receive $m$* $ 2f+1$ suggestions greater than $m*2f$, a contradiction.
     
\end{enumerate}

\end{proof}

The fourth proof assures that after the $suggest$ step, one or more proposals and provable broadcasts are common to $2f+1$ parties.




\subsection{Differed Figures}\label{App:Motivation Figure}

\begin{figure*}[ht]
     \centering
     \begin{subfigure}[b]{0.47\textwidth}
         \centering
         \includegraphics[width=\textwidth]{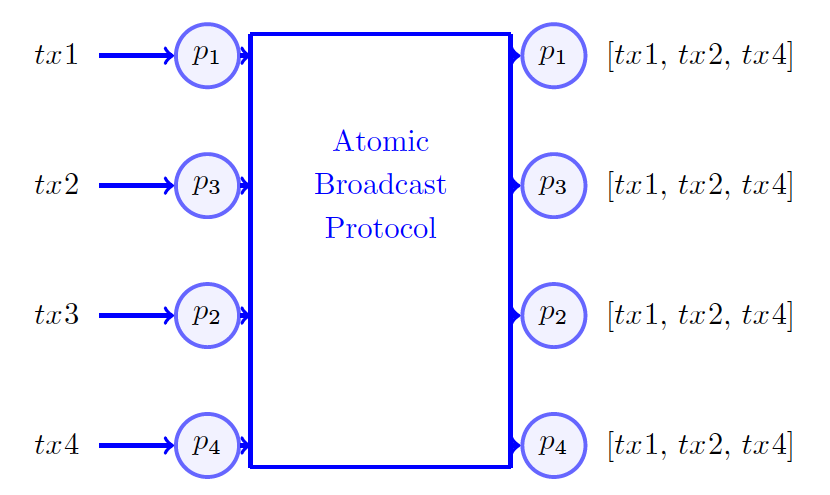}
         \caption{An ideal scenario where each party has unique transactions, and the parties reach an agreement on (n-f) parties' transactions and throughput is good. }
         \label{fig:Motivationfig1}
     \end{subfigure}
     \hfill
     \begin{subfigure}[b]{0.47\textwidth}
         \centering
         \includegraphics[width=\textwidth]{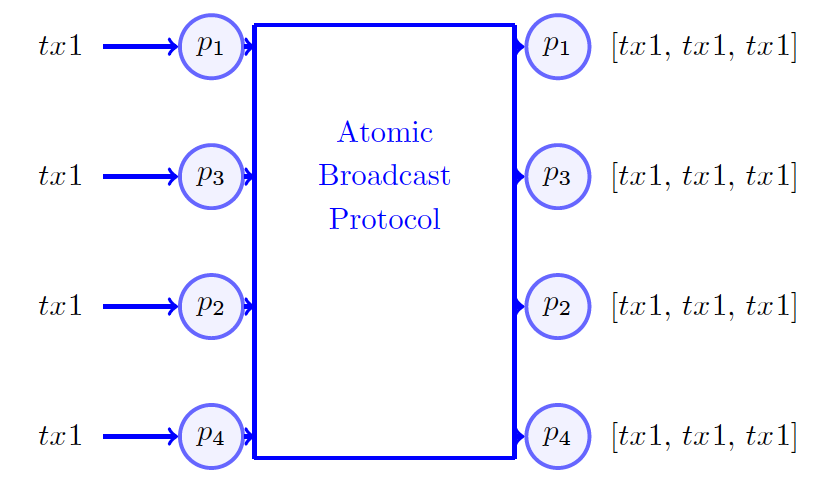}
         \caption{A scenario where parties have same transactions, therefore, though the parties agree on (n-f) parties' transactions, the throughput is not good.  }
         \label{fig:Motivationfig2}
     \end{subfigure} 
     \label{fig:Motivation1}
\end{figure*} 

\begin{figure*}[ht]
 \centering
    \begin{subfigure}[b]{0.47\textwidth}
         \centering
         \includegraphics[width=\textwidth]{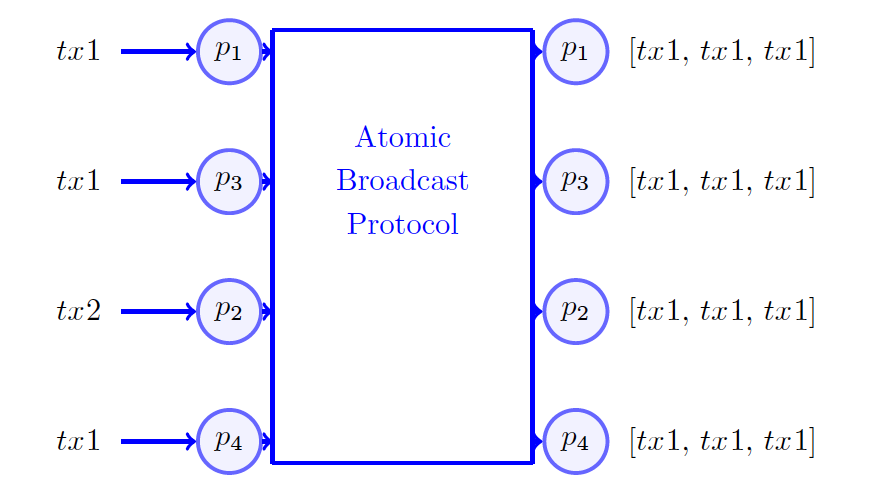}
         \caption{A scenario where (n-f) parties have same tranactions; therefore, though the parties agree on (n-f) parties' transactions, the throughput is not good.}
         \label{fig:Motivationfig3}
     \end{subfigure}
     \hfill
     \begin{subfigure}[b]{0.47\textwidth}
         \centering
         \includegraphics[width=\textwidth]{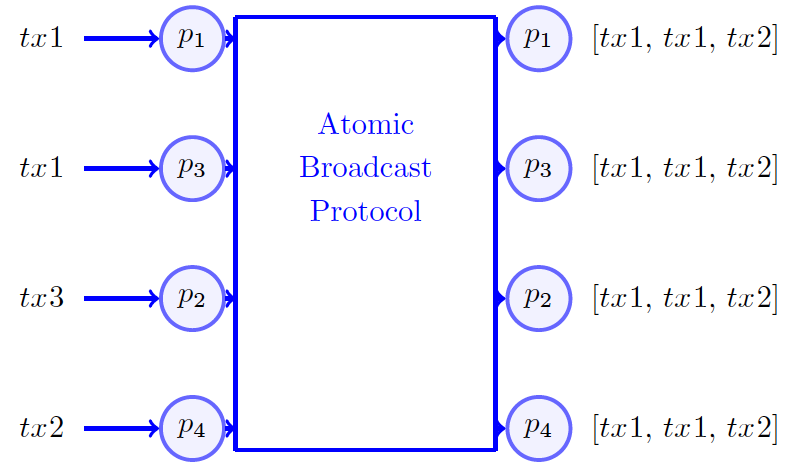}
         \caption{A scenario where there is a difference in the transactions among half of the parties; but there is duplication in agreed transactions. }
         \label{fig:Motivationfig4}
     \end{subfigure}
     \label{fig:Motivation2}
\end{figure*}

\begin{figure*}[ht]
     \centering
     \begin{subfigure}[b]{0.33\textwidth}
         \centering
         \includegraphics[width=\textwidth]{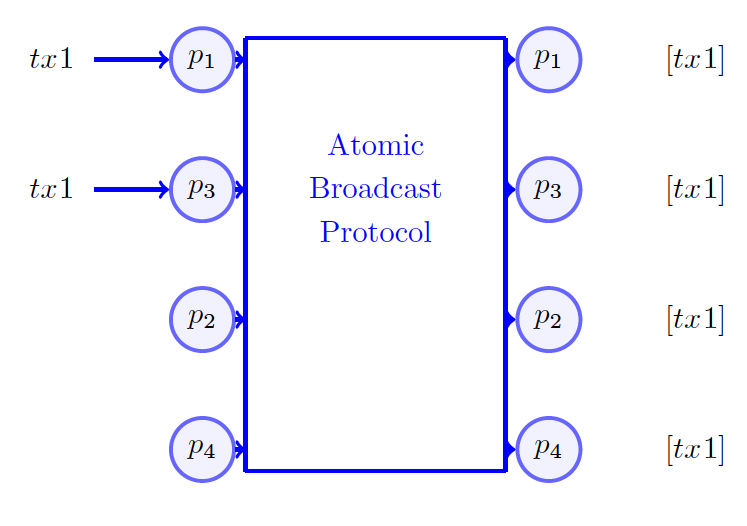}
         \caption{ A scenario where there is no difference in the transactions among the parties, therefore, though the parties agree on (f+1) parties' transactions, the throughput is not good, but the communication complexity is low.}
         \label{fig:Result1}
     \end{subfigure}
     \hfill
     \begin{subfigure}[b]{0.3\textwidth}
         \centering
         \includegraphics[width=\textwidth]{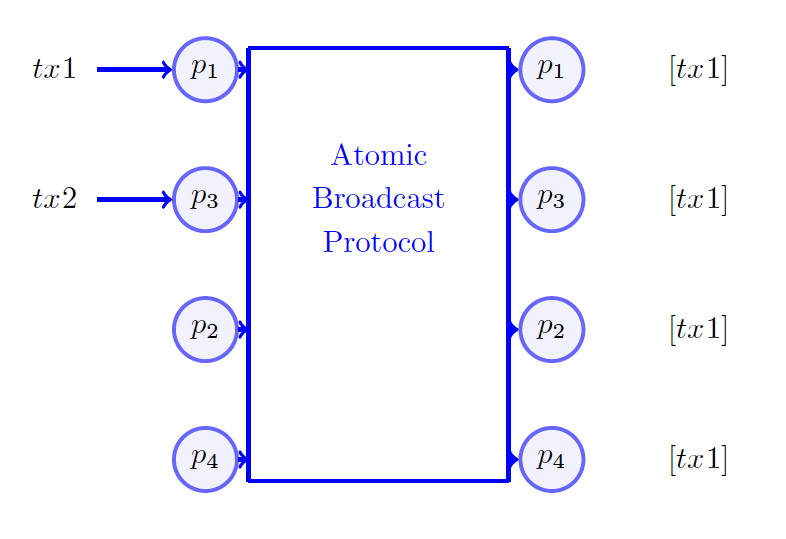}
         \caption{ A scenario where there is a difference in the transactions among the selected parties, but the parties reach an agreement on one party's requests. It is good because the complexity of communication is good.}
         \label{fig:Result2}
     \end{subfigure} 
      \hfill
     \begin{subfigure}[b]{0.32\textwidth}
         \centering
         \includegraphics[width=\textwidth]{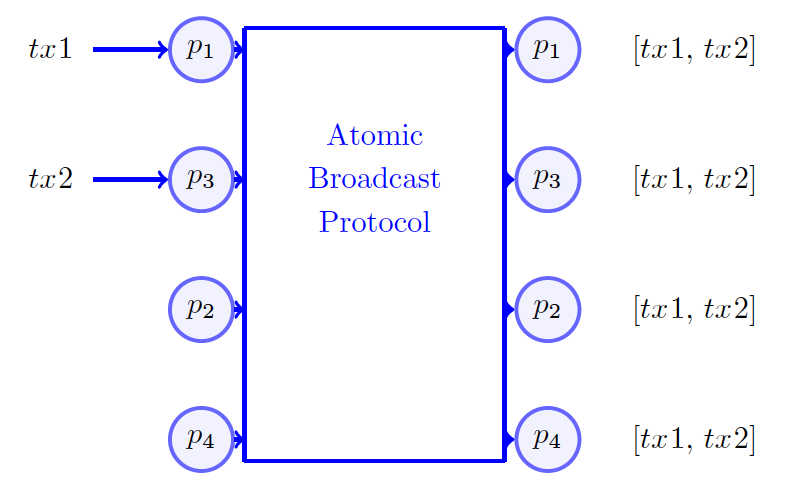}
         \caption{ A scenario where there are differences in the transactions among the parties. Therefore, though the parties agree on (f+1) parties' transactions, the throughput is good in low communication cost.}
         \label{fig:Result3}
     \end{subfigure} 
\end{figure*} 

\section{Related work on asynchronous and partially synchronous model.} \label{Related :Work}
\paragraph{Asynchronous settings} A protocol can reach a consensus on the BA problem if there are a total of $n$ parties, among them $f$ are faulty, and $n$ is at least greater than $3f$ \cite{BYZ10}. So any solution for the byzantine agreement problem has optimal resilience if it satisfies the following constraint, $n=3f+1$. Fischer, Lynch, and Paterson\cite{CONS03} gave a theorem that proved that byzantine agreement protocol does not have a termination property in asynchronous settings even if there is only one non-byzantine failure. Then Ben-or \cite{BYZ11} proved that in such situations, if we take the help of randomness, these protocols can terminate with a probability of almost $1$. The classic work of Cachin et al. \cite{SECURE02} presented asynchronous binary agreement (ABA), which is the building block of the MVBA protocol \cite{SECURE02}. MVBA allows every party to provide an input, and the protocol outputs one of the inputs. These inputs are externally valid by a predicate defined by the protocol. The protocol uses the threshold-signature scheme and coin-tossing 
 scheme \cite{THRESH01,BORN01} to realize the security and the randomness which is also used by the fault-tolerant protocols \cite{CACHIN01,VICTOR01,CACHIN02,BYZ20}.  Message complexity of the MVBA protocol  is $O(n^3)$, and it maintains optimal resilience. Recent work of Abraham et al. \cite{BYZ17} reduces the message complexity from $O(n^3)$ to $O(n^2)$ where the probability of a protocol terminates with a completed broadcast is $2/3$. Dumbo-MVBA \cite{BYZ20} also uses the MVBA as a base and achieves $O(n^2)$ message complexity but uses erasure code to minimize the message complexity.\;

\paragraph{Partially synchronous model} The partially synchronous communication model was introduced by Dwork, Lynch, and  Stockmeyer \cite{BYZ05}. This model assumes a known time bound $\Delta$ for message delay; that is, honest parties deliver their messages in this time bound after a \textit{global stabilization time (GST)}. After GST, a protocol advances deterministically \cite{CONS03}. \;

Castro et al. \cite{BYZ08} provided the first byzantine fault-tolerance protocol that assumes a partially synchronous model. The core of the protocol is a leader who receives requests from the clients, assigns orders on the received requests, and drives the other parties to reach a consensus. If a leader fails to deliver the result in $\Delta$ time-bound, then the parties start the leader election to elect a new leader. An adversary with the help of the byzantine parties can exploit this $\Delta$ parameter to drive the parties to find a new leader and makes the leader election process infinite \cite{HONEYBADGER01}. Many other protocols are proposed in the literature \cite{BYZ24,BYZ25,BYZ26,BYZ27,BYZ28,BYZ29,Hotstuff01} face the same challenges.

\end{document}